\documentclass{llncs}

\def\boxit#1{\vbox{\hrule\hbox{\vrule\kern4pt
  \vbox{\kern1pt#1\kern1pt}
\kern2pt\vrule}\hrule}}

\usepackage[tight,hang,nooneline]{subfigure}
 \usepackage{amsmath}
 \usepackage{amssymb}
\usepackage{psfrag}
\usepackage{here}
\usepackage{graphicx}
\usepackage{url} 
\usepackage{algorithm,myalgorithmic}
\renewcommand{\algorithmicrequire}{\textbf{Input(s):}}
\renewcommand{\algorithmicensure}{\textbf{Output(s):}}
\newcommand{\defines}{\emph}
\newcommand{\mc}{\mathcal}

\newcommand{\comment}[1]{ }

\newcommand{\YES}{\textup{\textsf{YES}}}
\newcommand{\NO}{\textup{\textsf{NO}}}
\newcommand{\Oh}{{\mc{O}}}

\renewcommand{\int}{\operatorname{int}}       

\newcommand{\sm}{\setminus}

\newcommand{\Pol}{\mbox{$\mc{P}$}}

\newcommand{\NP}{\mbox{$\mc{NP}$}}

\newcommand{\q}{{\cal Q}}
\newcommand{\p}{{\cal P}}
\newcommand{\qq}{\mathfrak{Q}}

\newcommand{\emphprob}[4]{{\bf {\sc  #1}} \\
{\bf Given:} #2, and the parameter #3.\\
{\bf We ask:} #4}

\begin{document}

\title{A Parameterized Perspective on $P_2$-Packings}
\author{Jianer Chen\inst{2,3}
      Henning Fernau\inst1 
       Dan Ning\inst2
      Daniel Raible\inst1
	Jianxin Wang\inst2
     }
\institute{%
Universit\"at 
Trier, FB IV---Abteilung Informatik, 
54286 Trier,
Germany,  
\\
\email{\{fernau,raible\}@informatik.uni-trier.de}
\and
School of Information Science and Engineering, 
Central South University, 
Changsha, Huan 410083, P.R. China,  
\email{alina1028@hotmail.com,jxwang@mail.csu.edu.cn}
\and
Department of Computer Science, 
Texas A\&M University, 
College Station,\\ Texas 77843, USA, 
\email{chen@cs.tamu.edu}
}

\maketitle
\begin{abstract}We study (vertex-disjoint) $P_2$-packings in graphs under a parameterized perspective. Starting from a maximal $P_2$-packing $\p$ of size $j$ we use extremal arguments for determining how many vertices of $\p$ appear in some $P_2$-packing of size $(j+1)$. We basically can 'reuse' $2.5j$ vertices. We also present a kernelization algorithm that gives a kernel of size
bounded by $7k$. With these two results we build an algorithm which constructs a $P_2$-packing of size $k$ in time $\Oh^*(2.482^{3k})$.
\end{abstract}

\section{Introduction and Definitions}

\paragraph{Motivation.}
We consider a generalization of the  matching problem in graphs. A matching is a maximum cardinality set of vertex disjoint edges. Our problem is a generalization in the sense that the term edge may be replaced by 2-edge-path called $P_2$. More formally, we study the following problem, called 
\emphprob{$P_2$-packing:}{A graph $G=(V,E)$}{$k$}{Is there a set of $k$ vertex-disjoint $P_2$'s in $G$?}\\[1ex]
P. Hell and D. Kirkpatrick~\cite{KirHel78,HelKir82} proved \NP-completeness for this problem. 
 {\sc $P_2$-packing} attracts attention as it is \NP-hard whereas the matching problem, which is {\sc $P_1$-packing}, is poly-time solvable.\\
Also there is a primal-dual relation to {\sc total edge cover} shown by H.~Fernau and D.~F. Manlove~\cite{FerMan06a}. 
Recall that an \defines{edge cover} is a set of edges  $EC \subseteq E$  that cover all vertices of a given graph $G=(V,E)$.
An edge cover is called \defines{total} if  every component in $G[EC]$ has at least two edges.
By matching techniques, the problem of finding an edge cover of size at most $k$ is poly-time solvable.
However, the following Gallai-type identity shows that finding total edge covers of size at most $k$ is
\NP-hard: The sum of the number of $P_2$'s in a maximum $P_2$-packing and the size of  a minimum total edge cover equals $n=|V|$.\\
There is also a relation to {\sc test cover}. The input to this problem is a hypergraph $H=(G,E)$ and one wishes to identify a subset $E' \subseteq E$ (the \defines{test cover}) such that, for any distinct $i,j \in V$, there is an $e' \in E'$ with $|e' \cap \{i.j\}|=1$. {\sc Test cover} models identification problems: Given a set of individuals and a set of binary attributes we search for a minimum subset of attributes that identifies each individual distinctly. Applications, as mentioned in K.~M.~J. Bontridder \emph{et al.}~\cite{Bonetal2003}, range from fault testing and diagnosis, pattern recognition to biological identification. K.~M.~J. Bontridder \emph{et al.} could show for the case TCP2, where  for all $e \in E$  we have $|e|\le 2$,  the subsequent two statements. First, if $H$ has a test cover of size $\tau$, 
then there is a $P_2$-packing of size $n-\tau-1$ that leaves at least one vertex isolated. Second, if $H$ has a maximal $P_2$-packing of size $\pi$ that leaves at least one vertex isolated, then there is a test cover of size $n-\pi-1$. This also establishes a close relation between {\sc test cover} and {\sc total edge cover}. 
So, we can employ 
our algorithms to solve the TCP2 case of {\sc test cover}, by using an initial
catalytic branch that determines one vertex that should be isolated.
 
\paragraph{Discussion of Related Work.}
 R.~Hassin and S.~Rubinstein~\cite{HasRub2006} found a randomized  $\frac{35}{67}$-approximation.
K.~M.~J. Bontridder \emph{et al.} \cite{Bonetal2003} studied $P_2$-packing also in the context of approximation, where 
they considered a series of heuristics $H_\ell$. $H_\ell$ 
 starts from a maximal $P_2$-packing $\p$ and attempts to improve it by replacing any $\ell$ $P_2$'s by $\ell+1$ $P_2$'s not contained yet in $\p$. The corresponding approximation ratios $\rho_\ell$ are as follows: $\rho_0=\frac{1}{3}$, $\rho_1=\frac{1}{2}$, $\rho_2=\frac{5}{9},\rho_3=\frac{7}{11}$ and $\rho_\ell=\frac{2}{3}$ for $\ell\geq 4$.
\\
As any {\sc $P_2$-packing} instance can be transformed to fit in  {\sc 3-set packing} one can use Y.~Liu \emph{et al.}~\cite{Liuetal2006} algorithm which needs $\Oh^*(4.61^{3k})$ steps.  
The first paper to individually study {\sc $P_2$-packing} under a parameterized view was E.~Prieto  and C.~Sloper~\cite{PriSlo2004}. The authors were able to prove a $15k$-kernel. Via a clever 'midpoint' search on the kernel they could achieve a run time of $\Oh^*(3.403^{3k})$.
 
\paragraph{Our Contributions.} The two main algorithmic achievements of this paper are: (1) a new $7k$-kernel for {\sc $P_2$-Packing},
(2) an algorithm which solves this problem in $\Oh^*(2.482^{3k})$. This algorithm makes use of a new theorem that says that we basically can reuse $2.5j$ vertices of a maximal $P_2$-packing with size $j$. This improves a similar result for general {\sc 3-set packing}~\cite{Liuetal2006} where only $2j$ elements are reusable. This theorem is proven by making extensive use of extremal combinatorial arguments. 
Another novelty is that in this algorithm,
 the dynamic programming phase (used
to inductively augment maximal $P_2$-packings) is interleaved with kernelization. 
This pays off not only heuristically but also asymptotically by a specific form of combinatorial analysis. Thereby we can completely skip the time consuming color-coding which was needed in Liu \emph{et al.}~\cite{Liuetal2006} for {\sc 3-set packing}. We believe that the idea of
saving colors by extremal combinatorial arguments could be applied in other situations, as well.

\paragraph{Some Notations and Definitions.} 
We only consider undirected graphs $G = (V, E)$. For a subgraph $H$ of $G$, denote by
$N(H)$ the set of vertices that are not in $H$ but adjacent to at
least one vertex on $H$, i.e., $N(H)=(\bigcup_{v\in H}N(\{v\}))\sm H$. 
The subgraph $H$ is {\it
adjacent} to a vertex $v$ if $v \in N(H)$. 
 A $P_2$ in $G$ is a path which consists of three vertices and two edges. For any path $p$ of this kind  we consider the vertices as numbered such that $p=p_1p_2p_3$ (where the roles of $p_1$ and $p_3$ might be interchanged). For a path $p$, $V(p)$ ($E(p)$, resp.) denotes the set of vertices (edges, resp.) on $p$. Likewise, for a set of paths $\p$, $V(\p):=\bigcup_{p\in\p}V(p)$ ($E(\p):=\bigcup_{p \in \p}E(p)$, resp.). 
\section{Kernelization}

We are going to improve on the earlier $15k$-kernel of E.~Prieto and C.~Sloper by allowing
local improvements on a maximal $P_2$-packing, but otherwise using the ideas of E.~Prieto and C.~Sloper.
Therefore, we first revise the necessary notions and lemmas from their paper~\cite{PriSlo2004}.

\subsection{Essential Prerequisites}

\begin{definition} 
A {\it double crown decomposition} of a graph $G$ is a decomposition
$(H, C, R)$ of the vertices in $G$ such that
\begin{enumerate}
 \item $H$ (the head) separates $C$ and $R$;
\item 
 $C = C_0 \cup C' \cup C''$ (the crown) is an independent set such
that $|C'| = |H|$, $|C''| = |H|$, and there exist a perfect matching
between $C'$ and $H$, and a perfect matching between $C''$ and $H$.
\end{enumerate}
\end{definition}


\begin{definition} 
A {\it fat crown decomposition} of a graph $G$ is a decomposition
$(H, C, R)$ of the vertices in $G$ such that
\begin{enumerate}
 \item $H$ (the head) separates $C$ and $R$;
\item the induced subgraph $G(C)$ is a collection of pairwise disjoint
$K_2$'s;
\item there is a perfect matching $M$ between $H$ and a subset of
vertices in $C$ such that each connected component in $C$ has at
most one vertex in $M$.
\end{enumerate}
\end{definition}

E.~Prieto and C.~Sloper could show the following lemmas.

\begin{lemma} 
\label{lem1}
 A graph $G$ with a double crown $(H, C, R)$ has a $P_2$-packing
 of size $k$ if and only if the graph $G - H -C$ has a
 $P_2$-packing of size $k - |H|$.
\end{lemma}

\begin{lemma} 
\label{lem2}
 A graph $G$ with a fat crown $(H, C, R)$ has a $P_2$-packing of
 size $k$ if and only if the graph $G - H -C$ has a $P_2$-packing
 of size $k - |H|$.
\end{lemma}


\begin{lemma} 
\label{lem3}
 A graph $G$ with an independent set $I$, where $|I| \geq 2 |N(I)|$,
 has a double crown decomposition $(H, C, R)$, where $H \subseteq
 N(I)$, which can be constructed in linear time.
\end{lemma}

\begin{lemma} 
\label{lem4}
 A graph $G$ with a collection $J$ of independent $K_2$'s, where
 $|J| \geq |N(J)|$, has a fat crown decomposition $(H, C, R)$,
 where $H \subseteq N(J)$, which can be constructed in linear time.
\end{lemma}


\subsection{A Smaller Kernel}
Let $G$ be a graph, and let ${\cal P} = \{L_1, \ldots, L_t\}$ be a
maximal $P_2$-packing in $G$, where each $L_i$ is a subgraph in $G$
that is isomorphic to $P_2$, and $t < k$. Then each connected
component of the graph $G - {\cal P}$ is either a single vertex or a
single edge. Let $Q_0$ be the set of all vertices such that each
vertex in $Q_0$ makes a connected component of $G - {\cal P}$ (each
vertex in $Q_0$ will be called a {\it $Q_0$-vertex}). Let $Q_1$ be
the set of all edges such that each edge in $Q_1$ makes a connected
component of $G - {\cal P}$ (each edge in $Q_1$ will be called a
{\it $Q_1$-edge}).

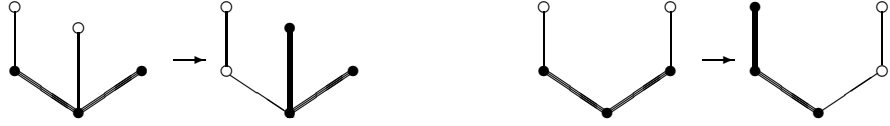
\begin{figure} \begin{center}
\setlength{\unitlength}{.8\unitlength}
\begin{picture}(410,60)
\put(0,20){\circle*{5}} \put(30,0){\circle*{5}}
\put(60,20){\circle*{5}} \put(0,20){\line(3,-2){30}}
\put(0,19){\line(3,-2){30}} \put(0,21){\line(3,-2){30}}
\put(30,0){\line(3,2){30}} \put(30,-1){\line(3,2){30}}
\put(30,1){\line(3,2){30}} \put(0,50){\circle{5}}
\put(30,40){\circle{5}} \put(0,48){\line(0,-1){30}}
\put(30,38){\line(0,-1){40}} \put(75,25){\vector(1,0){15}}

\put(100,20){\circle{5}} \put(130,0){\circle*{5}}
\put(160,20){\circle*{5}} \put(102,18){\line(3,-2){30}}
\put(130,0){\line(3,2){30}} \put(130,-1){\line(3,2){30}}
\put(130,1){\line(3,2){30}} \put(100,50){\circle{5}}
\put(130,40){\circle*{5}} \put(100,48){\line(0,-1){26}}
\put(130,38){\line(0,-1){40}} \put(129,38){\line(0,-1){40}}
\put(131,38){\line(0,-1){40}}

\put(250,20){\circle*{5}} \put(280,0){\circle*{5}}
\put(310,20){\circle*{5}} \put(250,20){\line(3,-2){30}}
\put(250,19){\line(3,-2){30}} \put(250,21){\line(3,-2){30}}
\put(280,0){\line(3,2){30}} \put(280,-1){\line(3,2){30}}
\put(280,1){\line(3,2){30}} \put(250,50){\circle{5}}
\put(310,50){\circle{5}} \put(250,48){\line(0,-1){30}}
\put(310,48){\line(0,-1){30}} \put(325,25){\vector(1,0){15}}

\put(350,20){\circle*{5}} \put(380,0){\circle*{5}}
\put(410,20){\circle{5}} \put(350,20){\line(3,-2){30}}
\put(350,19){\line(3,-2){30}} \put(350,21){\line(3,-2){30}}
\put(380,0){\line(3,2){28}} \put(350,50){\circle*{5}}
\put(410,50){\circle{5}} \put(350,48){\line(0,-1){30}}
\put(349,48){\line(0,-1){30}} \put(351,48){\line(0,-1){30}}
\put(410,48){\line(0,-1){25}}
\end{picture}

\caption{Reducing the number of $Q_0$-vertices} \label{fig1}
\end{center} \vspace*{-6mm}
\end{figure}

\begin{figure} \begin{center}
\setlength{\unitlength}{.8\unitlength}
\begin{picture}(410,60)
\put(0,20){\circle*{5}} \put(30,0){\circle*{5}}
\put(60,20){\circle*{5}} \put(0,20){\line(3,-2){30}}
\put(0,19){\line(3,-2){30}} \put(0,21){\line(3,-2){30}}
\put(30,0){\line(3,2){30}} \put(30,-1){\line(3,2){30}}
\put(30,1){\line(3,2){30}} \put(0,50){\circle{5}}
\put(30,40){\circle{5}} \put(20,50){\circle{5}}
\put(50,40){\circle{5}} \put(0,48){\line(0,-1){30}}
\put(30,38){\line(0,-1){40}} \put(3,50){\line(1,0){15}}
\put(33,40){\line(1,0){15}} \put(75,25){\vector(1,0){15}}

\put(100,20){\circle*{5}} \put(130,0){\circle*{5}}
\put(160,20){\circle{5}} \put(100,20){\line(3,-2){30}}
\put(130,0){\line(3,2){28}} \put(100,50){\circle*{5}}
\put(130,40){\circle*{5}} \put(120,50){\circle*{5}}
\put(150,40){\circle*{5}} \put(100,48){\line(0,-1){30}}
\put(99,48){\line(0,-1){30}} \put(101,48){\line(0,-1){30}}
\put(130,38){\line(0,-1){40}} \put(129,38){\line(0,-1){40}}
\put(131,38){\line(0,-1){40}} \put(103,50){\line(1,0){15}}
\put(103,49){\line(1,0){15}} \put(103,51){\line(1,0){15}}
\put(133,40){\line(1,0){15}} \put(133,39){\line(1,0){15}}
\put(133,41){\line(1,0){15}}

\put(250,20){\circle*{5}} \put(280,0){\circle*{5}}
\put(310,20){\circle*{5}} \put(250,20){\line(3,-2){30}}
\put(250,19){\line(3,-2){30}} \put(250,21){\line(3,-2){30}}
\put(280,0){\line(3,2){30}} \put(280,-1){\line(3,2){30}}
\put(280,1){\line(3,2){30}} \put(250,50){\circle{5}}
\put(270,50){\circle{5}} \put(310,50){\circle{5}}
\put(290,50){\circle{5}} \put(253,50){\line(1,0){15}}
\put(307,50){\line(-1,0){15}} \put(250,48){\line(0,-1){30}}
\put(310,48){\line(0,-1){30}} \put(325,25){\vector(1,0){15}}

\put(350,20){\circle*{5}} \put(380,0){\circle{5}}
\put(410,20){\circle*{5}} \put(350,20){\line(3,-2){27}}
\put(383,0){\line(3,2){27}} \put(350,50){\circle*{5}}
\put(370,50){\circle*{5}} \put(410,50){\circle*{5}}
\put(390,50){\circle*{5}} \put(353,50){\line(1,0){15}}
\put(353,49){\line(1,0){15}} \put(353,51){\line(1,0){15}}
\put(407,50){\line(-1,0){15}} \put(407,49){\line(-1,0){15}}
\put(407,51){\line(-1,0){15}} \put(350,48){\line(0,-1){30}}
\put(349,48){\line(0,-1){30}} \put(351,48){\line(0,-1){30}}
\put(410,48){\line(0,-1){30}} \put(409,48){\line(0,-1){30}}
\put(411,48){\line(0,-1){30}}
\end{picture}
\caption{Reducing the number of $Q_1$-edges} \label{fig2}
\end{center} \vspace*{-6mm}
\end{figure}
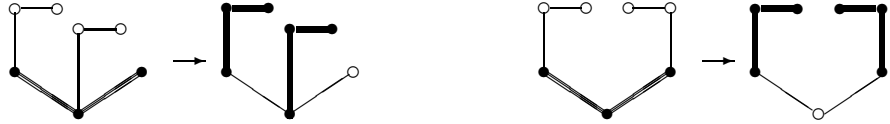

Our kernelization algorithm starts with the following process, which
tries to reduce the number of $Q_0$-vertices and the number of
$Q_1$-edges, by applying the following rules:\footnote{We have used
solid circles and thick lines for vertices and edges, respectively,
in the $P_2$-packing $\cal P$, and hollow circles and thin lines for
vertices and edges not in $\cal P$. In particular, two hollow circles
linked by a thin line represents a $Q_1$-edge.}
\begin{quote}
 {\bf Rule 1.} If a $P_2$-copy $L_i$ in $\cal P$ has two vertices
   that each is adjacent to a different $Q_0$-vertex, then apply
   the processes described in Figure~\ref{fig1} to decrease the number
   of $Q_0$-vertices by 2 (and increase the number of $Q_1$-edges by 1).

 {\bf Rule 2.} If a $P_2$-copy $L_i$ in $\cal P$ has two vertices
   that each is adjacent to a different $Q_1$-edge, then apply the
   processes described in Figure~\ref{fig2} to decrease the number of
   $Q_1$-edges by 2 (and increase the size of the maximal $P_2$-packing
   by 1).
\end{quote}
Note that these rules cannot be applied forever. The number of
consecutive applications of Rule 1 is bounded by $n/2$ since each
application of Rule 1 reduces the number of $Q_0$-vertices by 2; and
the total number of applications of Rule 2 is bounded by $k$ since
each application of Rule 2 increases the number of $P_2$-copies in
the $P_2$-packing by 1. We also remark that during the applications
of these rules, the resulting $P_2$-packing $\cal P$ may become
non-maximal. In this case, we simply first make $\cal P$ maximal
again, using any proper greedy algorithm, before we further apply
the rules.

Therefore, we must reach a point, in polynomial time, where none of
the rules above is applicable. At this point, the maximal
$P_2$-packing $\cal P$ has the following properties:
\begin{quote}
 For each $L_i$ of the $P_2$-copies in $\cal P$,

 {\bf Property 1.} If more than one $Q_0$-vertices are adjacent
  to $L_i$, then all these $Q_0$-vertices must be adjacent
  to the same (and unique) vertex in $L_i$.

 {\bf Property 2.} If more than one vertex in $L_i$ are adjacent
 to $Q_0$-vertices, then all these vertices in $L_i$ must be
 adjacent to the same (and unique) $Q_0$-vertex.

 {\bf Property 3.} If more than one $Q_1$-edges are adjacent
  to $L_i$, then all these $Q_1$-edges must be adjacent
  to the same (and unique) vertex in $L_i$.

 {\bf Property 4.} If more than one vertex in $L_i$ are adjacent
 to $Q_1$-edges, then all these vertices in $L_i$ must be
 adjacent to the same (and unique) $Q_1$-edge.
\end{quote}


\begin{theorem} \label{thm1}
 Let ${\cal P} = \{L_1, L_2, \ldots, L_t\}$ be a maximal $P_2$-packing
 on which Rules 1-2 are not applicable, where $t \leq k-1$. If the
 number of $Q_0$-vertices is larger than $2k-3$, then there is a
 double crown, 
constructible in linear time.
\end{theorem}

\begin{proof}
We partition the $Q_0$-vertices into two groups: the group $Q_0'$
that consists of all the $Q_0$-vertices such that each $Q_0$-vertex
in $Q_0'$ has at least two different neighbors in a single
$P_2$-copy $L_i$ of $\cal P$; and $Q_0'' = Q_0 - Q_0'$. Without loss
of generality, let ${\cal L}_1 = \{L_1, \ldots, L_d\}$ be the
collection of $P_2$-copies in $\cal P$ such that each $L_i$ in
${\cal L}_1$ has at least two vertices that are adjacent to the same
vertex in $Q_0'$. By Property 2, at most one vertex in $Q_0'$ can be
adjacent to each $P_2$-copy $L_i$ in ${\cal L}_1$. Therefore,
$|Q_0'| \leq |{\cal L}_1| = d$, which also implies that $Q_0''$ is
not empty. Moreover,
  \[  |Q_0''| = |Q_0| - |Q_0'| \geq 2k-2-d \geq 2(k-1-d)
      \geq 2t - 2d.  \]
By property 2 again, no vertex in $Q_0''$ can be adjacent to any
$P_2$-copy $L_i$ in ${\cal L}_1$. Therefore, the neighbors of the
vertices in $Q_0''$ are all contained in the collection ${\cal L}_2
= \{L_{d+1}, \ldots, L_t\}$. By definition, there is at most one
vertex in each $L_i$ in ${\cal L}_2$ that is adjacent to
$Q_0$-vertices. Therefore, the total number $|N(Q_0'')|$ of
neighbors of $Q_0''$ is bounded by $t - d$. Note that $Q_0''$ is an
independent set, and
  \[ |Q_0''| \geq 2t - 2d = 2(t-d) \geq 2 \cdot |N(Q_0'')|  \]
By Lemma~\ref{lem3}, the graph has a double crown that can be
constructed in linear time.
\qed\end{proof}

The proof of the following theorem is quite similar to that of
Theorem~\ref{thm1}. 

\begin{theorem} \label{thm2}
 Let ${\cal P} = \{L_1, L_2, \ldots, L_t\}$ be a maximal $P_2$-packing
 on which Rules 1-2 are not applicable, where $t \leq k-1$. If the
 number of $Q_1$-edges is larger than $k-1$, then there is a
 fat crown, which can be constructed in linear time.
\end{theorem}

\begin{proof}
We partition the $Q_1$-edges into two groups: the group $Q_1'$ that
consists of all the $Q_1$-edges such that each $Q_1$-edge in $Q_1'$
has at least two different neighbors in a single $P_2$-copy $L_i$ of
$\cal P$; and $Q_1'' = Q_1 - Q_1'$. Without loss of generality, let
${\cal L}_1 = \{L_1, \ldots, L_d\}$ be the collection of
$P_2$-copies in $\cal P$ such that each $L_i$ in ${\cal L}_1$ has at
least two vertices that are adjacent to the same edge in $Q_1'$. By
Property 4, at most one edge in $Q_1'$ can be adjacent to each
$P_2$-copy $L_i$ in ${\cal L}_1$. Therefore, $|Q_1'| \leq |{\cal
L}_1| = d$, which also implies that $Q_1''$ is not empty. Moreover,
  \[  |Q_1''| = |Q_1| - |Q_1'| \geq k-1-d \geq t - d.  \]
By property 4 again, no edge in $Q_1''$ can be adjacent to any
$P_2$-copy $L_i$ in ${\cal L}_1$. Therefore, the neighbors of the
edges in $Q_1''$ are all contained in the collection ${\cal L}_2 =
\{L_{d+1}, \ldots, L_t\}$. By definition, there is at most one
vertex in each $L_i$ in ${\cal L}_2$ that is adjacent to
$Q_1$-edges. Therefore, the total number $|N(Q_1'')|$ of neighbors
of $Q_1''$ is bounded by $t - d$. Note that $Q_1''$ is a collection
of independent $K_2$'s, and
  \[ |Q_1''| \geq t - d  \geq |N(Q_1'')|  \]
By Lemma~\ref{lem4}, the graph has a fat crown that can be
constructed in linear time.
\qed\end{proof}

Based on all these facts, our kernelization algorithm goes like
this: we start with a maximal $P_2$-packing $\cal P$, and repeatedly
apply Rules 1-2 (and keeping $\cal P$ maximal) until neither Rule 1
nor Rule 2 is applicable. At this point, if the number of
$Q_0$-vertices is larger than $2k-3$, then by Theorem~\ref{thm1}, we
generate a double crown that, by Lemma~\ref{lem1}, leads to a larger
$P_2$-packing. On the other hand, if the number of $Q_1$-edges is
larger than $k-1$, then by Theorem~\ref{thm2}, we generate a fat
crown that, by Lemma~\ref{lem2}, leads to a larger $P_2$-packing. By
repeating this process polynomial many times, either we will end up
with a $P_2$-packing of size at least $k$, or we end up with a
maximal $P_2$-packing $\cal P$ of size less than $k$ on which
neither Rule 1 nor Rule 2 is applicable, the number of
$Q_0$-vertices is bounded by $2k-3$, and the number of $Q_1$-edges
is bounded by $k-1$ (which implies that there are at most $2k-2$
vertices in $Q_1$). The vertices in the sets $Q_0$ and $Q_1$, plus
the at most $3k-3$ vertices in the $P_2$-packing, give a graph of at
most $7k - 8$ vertices.


\begin{theorem}\label{thm-p2kernel}
{\sc $P_2$-packing} admits a kernel with at most $7k$ vertices. 
\end{theorem}

We mention here that the (more general results) of H.~Fernau and D.~Manlove~\cite{FerMan06a}
can be improved for the parametric dual (in the sense of the mentioned Gallai-type identity) {\sc total edge cover}:

\begin{theorem}\label{thm-tec-kernel}
 {\sc total edge cover} admits a kernel with at most $1.5k_d$ vertices.
\end{theorem}

\begin{proof}
 Since we aim at a total edge cover, the largest number of vertices that can be covered by $k$ edges is $1.5k$ (namely,
if the edge cover is a $P_2$-packing). Hence, if the graph contains more than $1.5k$ vertices, we can reject.
This leaves us with a kernel with at most $1.5k$ vertices.
\qed\end{proof}


\begin{corollary}
Trivially, {\sc $P_2$-packing} does not admit a kernel with less than $3k$ vertices.
  {\sc total edge cover}  does not admit a kernel with less than $\alpha_ak_d$ vertices for any
$\alpha_d<(7/6)$, unless $\Pol=\NP$.
\end{corollary}

\begin{proof}
A $P_2$-packing of size $k$ is only possible in a graph with at least $3k$ vertices.
Due to Theorem~\ref{thm-p2kernel} and \cite[Theorem~3.1]{Cheetal2007},
there does not exist a kernel of size $\alpha_dk_d$ for   {\sc total edge cover} 
under the assumption that $\Pol=\NP$ if $(7-1)(\alpha_d-1)<1$.
 \qed
\end{proof}

\section{Combinatorial Properties of $P_2$-packings}

We consider the following setting. Let $\p$ be a maximal $P_2$-packing of size $j$ of a given graph $G=(V,E)$.
We will argue in this section that, whenever a $P_2$-packing of size $(j+1)$ exists, then there is also one,
called $\q$, that uses at least $2.5j$ out of the $3j$ vertices of $\p$. This combinatorial property of $\q$ 
(among others) will be used in the next section by the inductive step of our algorithm for {\sc $P_2$-packing}.
We employ extremal combinatorial arguments to achieve our results, deriving more and more properties that
$\q$ could possess, without risking to miss any $P_2$-packing of size $(j+1)$. 

So, among all $P_2$-packings of size $(j+1)$, we will 
consider those packings ${\cal Q}$  that maximize 
\begin{equation}\label{prop1}
\sum_{p \in \p}\sum_{q \in \q} 1_{[E(p)=E(q)]},
\end{equation} where $1_{[ \ ]}$ is the indicator function.
We call these $\qq_{(1)}$. In $\qq_{(1)}$ we find those packings from $\q$ who 'reuse' the maximum number of $P_2$'s in $\p$. From Liu et \emph{al.}~\cite{Liuetal2006} we know:
\begin{lemma}
 $|V(p) \cap V(\q)|\ge 2$ for any $p \in \p$ and $\q \in \qq_{(1)}$.
\end{lemma}

\begin{proof}
If there is $p\in{\cal P}$ with $|V(p)\cap V({\cal Q})|=1$, then replace the intersecting path of $\q$ by $p$ .
In the case where $|V(p) \cap V({\cal Q})|=0$, simply replace an arbitrary $q \in \q \setminus \p$, that must exist by pigeon-hole, by $p$.
In both cases, we obtain a packing $\q'$ of the same size as $\q$, but $|{\cal P}\cap {\cal Q}'|=|{\cal P}\cap {\cal Q}|+1$, contradicting 
$\q \in \qq_{(1)}$.
%
\qed\end{proof}
 A slightly sharper version is the next assertion:
\begin{corollary}\label{2pfade}
If $\q\in\qq_{(1)}$, then 
for any $p \in \p$ with $p \not \in \q$, 
there are $q_1,q_2 \in \q$ with $|V(p) \cap V(q_i)|\ge 1$ ($i=1,2$).
\end{corollary}
\begin{proof}
Suppose it exists $p \in \p$ and only one $q \in \q$ with $|V(p) \cap V(q)| \ge 2$.
Then $\q \setminus \{q\} \cup\{p\}$ improves on priority~(\ref{prop1}), contradicting  $\q\in\qq_{(1)}$.
\qed\end{proof}
 Furthermore, from the set $\qq_{(1)}$ we only collect those $P_2$-packings $\q'$, which maximize the following second property:\
\begin{equation}\label{prop2}
\sum_{p\in\mc{P}}\sum_{q\in\mc{Q'}}|E(p) \cap E(q)|.
\end{equation}
The set of the remaining $P_2$-packings will be called \defines{$\qq_{(2)}$}.
So, in $\qq_{(2)}$ are those packings who cover the maximum number of edges in $E(\p)$.
We define $\mc{P}_i(\mc{Q}):=\{p\in\mc{P}\mid i=|p\cap V(\mc{Q})|\}$.
A vertex $v \in V$ is a \defines{$\q$-endpoint} if there is a unique $q=q_1\ldots q_3 \in Q$ such that $v=q_1$ or $v=q_3$.
A vertex $v$ is called \defines{$\q$-midpoint} if there is a $q=q_1q_2q_3 \in \q$ with $q_2=v$.
\begin{definition}\
\begin{enumerate}
\item We call $q=q_1q_2q_3 \in \q$ \defines{foldable} on $p=p_1 p_2 p_3 \in \p$ if, for $q_2 \in V(p) \cap V(q)$, we have $p_s=q_2$, $s \in \{1,2,3\}$,  and either $p_{s+1}\not\in V(\q)$ or  $p_{s-1} \not\in V(\q)$, see Figure~\ref{foldfig1}.
\end{enumerate}
\begin{figure}
\centering
\psfrag{ar}{$\leadsto$}
\psfrag{q1}{$q_1$}
\psfrag{p1}{$p_1$}
\psfrag{q3}{$q_3$}
\subfigure[$q$ is foldable on~$p$.]{\label{foldfig1}\includegraphics[scale=0.85]{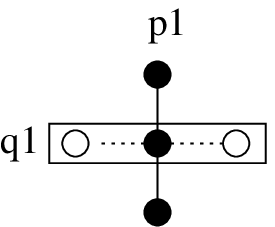}\hspace{3ex}}\hspace{3ex}
\subfigure[$(q_1,p_1)$-folding.]{\label{foldfig2}\includegraphics[scale=0.85]{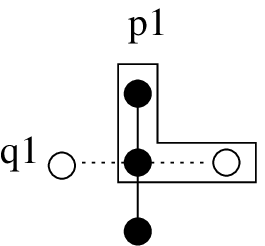}\hspace{3ex}}\hspace{3ex}
\subfigure[$q$ is shiftable  on~$p$.]{\label{shiftfig1}\includegraphics[scale=0.85]{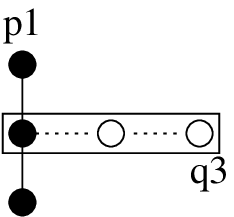}\hspace{3ex}}\hspace{3ex}
\subfigure[$(q_3,p_1)$-shifting.]{\label{shiftfig2}\includegraphics[scale=0.85]{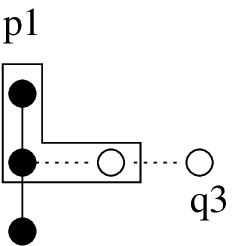}\hspace{3ex}}
\caption{The black vertices and solid edges indicate the $P_2$-packing $\p$. The polygons contain the $P_2$'s of the packing $\q$.}
\end{figure}
\begin{enumerate}
\item[2.] If $q$ is foldable on $p$, then substituting $q$ by $q \setminus \{q_i\} \cup\{ p_{s\pm 1}\}$ with $ i \in \{1,3\}$, will be called \defines{$(q_i,p_{s\pm 1})$-folding}, see Figure~\ref{foldfig2}.

\item[3.] We call $q=q_1q_2q_{3} \in \q$ \defines{shiftable} with respect to $q_1$ ($q_3$, resp.) on $p=p_1 p_2 p_3\in \p$ if the following holds: $q_1\in  V(p) \cap V(q)$  ( $q_3 \in V(p) \cap V(q)$, resp.) and either $p_{s+1}\not\in V(\q)$ or $p_{s-1} \not \in V(\q)$ where $p_s =q_1$ ($p_s = q_3$, resp.) and $s \in \{1,2,3\}$, see Figure~\ref{shiftfig1}.
\item[4.] If $q$ is shiftable on $p$ with respect to $t \in \{q_1,q_3\}$, then substituting $q$ by $q \setminus \{g\} \cup \{p_{s+1}\}$ (or by  $q \setminus \{g\} \cup \{p_{s-1}\}$, resp.), $g \in \{q_1,q_3\} \setminus \{t \}$, will be called \defines{$(g,p_{s+1})$-shifting} (\defines{$(g,p_{s-1})$-shifting}, resp.), see Figure~\ref{shiftfig2}.
\end{enumerate}
\end{definition}

\begin{lemma}\label{shiftable}
If $q=q_1q_2q_3 \in \q$ with $\q \in \qq_{(2)}$ is shiftable on $p \in \p$ with respect to $q_1$ (or $q_{3}$, resp.), then there is some $p' \in \p$ with $p' \neq p$ such that $\{q_3, q_2\} \in E(p)$ (or $\{q_2, q_1\} \in E(p)$, resp.).
\end{lemma}
\begin{proof}We examine the case where $V(p) \cap V(q)=\{q_1\}$ and, w.l.o.g., $p_{s+1} \not \in V(\q)$. Now assume the contrary. Then by $(q_3,p_{s+1})$-shifting, we obtain a $P_2$-packing $\q'$. Comparing $\q$ and $\q'$ with respect to priority~\ref{prop1},  $\q'$ is no worse than $\q$. But $\q'$ improves on priority~\ref{prop2}, as we gain $\{p_s,p_{s+1}\}$. But this contradicts $\q \in \qq_{(2)}$.
The case for $V(p)\cap V(q)=\{q_{3}\}$  follows analogously.
\qed\end{proof}

\begin{lemma}\label{nofold}
If $\q \in \qq_{(2)}$, then no $q \in \q$ is  foldable.
\end{lemma}
\begin{proof}
Suppose some  $q \in \q$ is foldable on $p$ and, w.l.o.g., $p_{s+1} \not\in V(\q)$ and $q_1 \not \in V(\p)$. Then by $(q_1,p_{s+1})$-folding $q$ we could improve on priority~\ref{prop2}, contradicting $\q \in \qq_{(2)}$.
\qed\end{proof}
Suppose there is a path $p$ with  $|V(p) \cap V(\q)|=2$. Then $p$ shares exactly one vertex $p_{q'},p_{q''}$ with paths $q',q'' \in \q$ due to Corollary~\ref{2pfade}.
In the following  $p_{q'}$ and $p_{q''}$ will always refer to the two cut vertices of the paths $q',q'' \in \q$ which cut a path $p$ with  $|V(p) \cap V(\q)|=2$.

\begin{lemma}\label{2fold}
Let $\q \in \qq_{(2)}$. 
Consider $p \in \p$ with $|V(p) \cap V(\q)|=2$ and neither $p_{q'}$ nor $p_{q''}$ are $\q$-endpoints. Then one of $q',q''$ is foldable.
\end{lemma}

\begin{proof}
Let $i,j \in \{1,2,3\}$ such that $p_{q'}=p_i$ and $p_{q''}=p_j$. Then for $f \in \{1,2,3\} \setminus \{i,j\}$, 
we have $p_f \not \in V(\q)$. W.l.o.g., $\{p_i ,p_f\} \in E(p)$. Then $q'$ is  $(q'_1,p_f)$-foldable.
\qed\end{proof}

\begin{corollary}\label{endpoint}
Let $\q \in \qq_{(2)}$ 
and $p \in \p$ with $|V(p) \cap V(\q)|=2$.  Then one of $p_{q'}, p_{q''}$ must be a $\q$-endpoint. 
\end{corollary}

\begin{proof}
Assume the contrary.
Lemmas~\ref{nofold} and~\ref{2fold} lead to a contradiction.
\qed\end{proof}

\begin{theorem}\label{p2}
Let $\p$ be a maximal $P_2$-packing of size $j$. If there is a $P_2$-packing of size $(j+1)$, then there is also a packing $\q \in  \qq_2$ such that  $|V(\p) \cap V(\q)|\ge 2.5j$.
\end{theorem}
\begin{proof}Suppose there is a path $p \in \p$ with $|V(p) \cap V(\q)|=2$. By Corollary~\ref{endpoint}, w.l.o.g.. $p_{q'}$ is a $\q$-endpoint. For $p_{q''}$ there are two possibilities: {\bf a)} $p_{q''}$ is also a $\q$-endpoint. Let $\{ p_f\}=V(p)\setminus \{p_{q'},p_{q''}\}$. Then w.l.o.g. $p_{q'}$ is a path neighbor of $p_f$. Therefore $p_{q'}$ is shiftable. {\bf b)} $p_{q''}$ is a $\q$-midpoint.\\
{\it Claim.} $p_{q''} \neq p_2$.\\ Suppose the contrary. Then w.l.o.g., $p_{q'}=p_1$ and thus $q''$ is foldable on $p$ by a $(q''_{1},p_3)$-folding. This contradicts Lemma~\ref{nofold}. The claim follows.\\
 W.l.o.g., we assume $p_{q''}=p_1$. Then it follows that $p_{q'}=p_2$, as otherwise a $(q''_{1},p_2)$-folding would contradict Lemma~\ref{nofold} again.\\
From $p_{q'}=p_2$ and $p_3 \not \in V(\q)$ we can derive that also in this case $p_{q'}$ is shiftable.\\[1ex]
We now examine for both cases the implications of the shiftability of $p_{q'}$. 
W.l.o.g., we suppose that $p_{q'}=q'_1$. 
Due to Lemma~\ref{shiftable} there is a $p' \in \p$ with $\{q'_3,q'_2\} \in E(p')$. 
From Corollary~\ref{2pfade}, it follows that there must be a $\bar{q} \in \q \setminus \{q'\}$ with $|V(p') \cap V(\bar{q})|=1$. 
Hence, $|V(p')\cap V(\q)|=3$. 
Note that $q'$ is the only path in $\q$ with $|V(q') \cap V(p')|=2$. 
Summarizing, we can say that for any $p \in \p$ with $|V(p) \cap V(\q)|=2$ we find a distinct $p' \in \p$ (via $q'$) such that $|V(p') \cap V(\q)|=3$. So, there is a total injection $ \gamma$ from $\p_2(\q)$ to $\p_3(\q)$. From $|\p_2(\q) \cup \p_3(\q)|=j$ and the existence of $\gamma$ we derive $|\p_2(\q)|\le 0.5j$. This implies $|V(\p) ¸\cap V(\q)|=2|\p_2(\q)|+3|\p_3(\q)| \ge 2.5j$.
\qed\end{proof}

\section{The Algorithm}

We like to point out the following two facts about $P_2$-packings. First, if a graph has a $P_2$-packing $\p=\{p^1, \ldots ,p^k\}$, then it suffices to know the set of midpoints ${\cal M}_{\p}=\{p^1_2,\ldots ,p^k_2\}$ to construct a $P_2$-packing of size $k$ (which is possibly $\p$) in poly-time. This fact was discovered by E. Prieto and C. Sloper~\cite{PriSlo2004} and basically can be achieved by bipartite matching techniques. Second, it also suffices to know the set of endpoint pairs $E_{\p}=\{(p^1_1,p^1_3),\ldots , (p^k_1,p^k_3) \}$ to construct a $P_2$-packing of size $k$ in poly-time.  This is due to Lemma 3.3 of Jia \emph{et al.}~\cite{JiaZhaChe2004} as any $P_2$-packing instance also can be viewed as a {\sc 3-set packing} instance.

\begin{algorithm}
\caption{An Algorithm for {\sc $P_2$-packing}.}
\label{algo1}
\begin{algorithmic}[1]
\STATE $\p = \emptyset$.
\STATE Greedily augment $\p$ to a maximal $P_2$-packing.
\STATE Apply {\bf Rule 1} and {\bf Rule 2} exhaustively and call the resulting packing $\p$.
\IF{There is a fat crown or double crown $(C,H,R)$ with $C \subseteq V \setminus V(\p)$}
\STATE $k \gets k -|H|$. $G \gets G -H -C$.
\STATE {\bf Goto} 1.
\ELSIF{$k\le 0$}
\STATE {\bf return} \YES
\ELSE
\STATE  Try to construct a $P_2$-packing $\p'$ From $\p$ with $|\p|+1=|\p'|$ using Algorithm~\ref{algo2}.
\IF{Step 10 failed}
\STATE {\bf return } \NO.
\ELSE
\STATE $\p \gets \p'$.
\STATE {\bf Goto} 2.
\ENDIF
\ENDIF
\end{algorithmic}
\end{algorithm}

\begin{algorithm}
\caption{An Algorithm for augmenting a maximal $P_2$-packing $\p$.}
\label{algo2}
\begin{algorithmic}[1]
\STATE $j  \gets |\p|$.
\FOR{$\ell$=0 to $0.3251j $}
\FORALL{$S_{i} \subseteq V(\p)$, $S_o \subseteq V \setminus V(\p)$ with $|S_{i}|=(j+1)-\ell$ and  $|S_o|=\ell  $}
\STATE Try to construct a $P_2$-packing $\p'$ with  $S_i \cup S_o$ as midpoints.
\IF{Step 4 succeeded}
\STATE {\bf return} $\p'$.
\ENDIF
\ENDFOR
\ENDFOR
\FOR{$\bar{\ell}=0$ to $0.1749j+3$}
\FORALL{$B_{i} \subseteq V(\p)$, $B_o \subseteq V \setminus V(\p)$ with $|B_{i}|=2(j+1)-\bar{\ell}$ and  $|B_o|=\bar{\ell} $}
\FORALL{possible endpoint pairs $(e^1_1,e^1_2),\ldots , (e^{j+1}_1,e^{j+1}_2)$ from $B_{i} \cup B_o$}
\STATE Try to construct a $P_2$-packing $\p'$with $(e^1_1,e^1_2),\ldots , (e^{j+1}_1,e^{j+1}_2)$ as endpoint pairs.
\IF{Step 10 succeeded}
\STATE {\bf return} $\p'$.
\ENDIF
\ENDFOR
\ENDFOR
\ENDFOR
\STATE {\bf return } failure.
\end{algorithmic}
\end{algorithm}

\subsection{Correctness}

\paragraph{Algorithm~\ref{algo1}.}
Steps one to six of Algorithm~\ref{algo1} are used for finding fat and double crowns. First we build a maximal $P_2$-packing and locally improve it via {\bf Rule 1} and {\bf Rule 2}. If afterwards we have $|Q_0|>2j-3$ we construct a double crown where $j:=|\p|$. If $|Q_1|>j-1$ we find a fat crown. These two actions are directly justified by Lemmas~\ref{lem1} and~\ref{lem2}. If we do not succeed anymore in finding either one of the two crown types we immediately can  rely on $|V(G)|\le 7j$. The next step  tries to construct a new $P_2$-packing $\p'$ from $\p$ such that $\p'$ comprises one more $P_2$ than $\p$. For this we invoke Algorithm~\ref{algo2}.

\paragraph{Algorithm~\ref{algo2}.}
If a $P_2$-packing $\p'$ with $|\p'|=j+1$ exists   we can  partion the midpoints $\mc{M}_{\p'}$ in a part which lies within $V(\p)$ and one which lies outside. We call them $\mc{M}^i_{\p'}:=\mc{M}_{\p'} \cap V(\p)$ and $\mc{M}^o_{\p'}:=\mc{M}_{\p'} \cap O$, respectively with $O:=V(\p') \setminus V(\p)$. Theorem~\ref{p2} yields $|O|\le 0.5j+3$ and thus  $|\mc{M}^o_{\p'}|\le 0.5j+3$. Basically, we can find an integer $\ell$ with $0 \le \ell \le 0.5j+3$ such that $|\mc{M}^i_{\p'}|=(j+1)-\ell$ and $|\mc{M}^o_{\p'}|= \ell$.\\
In step 3 we run through every such $\ell$ until we reach $0.3251j$. For any choice of $\ell$ in step 4 we cycle through all possibilities of choosing  sets $S_{i} \subseteq V(\p)$ and  $S_o \subseteq V \setminus V(\p)$ such that  $|S_{i}|=(j+1)-\ell$ and  $|S_o|=\ell$. Here $S_i$ and $S_o$ are candidates for $\mc{M}^i_{\p'}$ and $\mc{M}^o_{\p'}$, respectively. For any choice of $S_i$ and $S_o$ we try to construct a $P_2$-packing. If we succeed once we can return the desired larger $P_2$-packing. Otherwise we reach the point where $\ell=0.3251j$. At this point we change our strategy. Instead of looking for the midpoints of $\p'$ we focus on the endpoints. We do so because this will improve the run time as we will see later.  $O$ is the disjoint union of $\mc{M}^o_{\p'}$ and the endpoints of $\p'$ which do not lie in $V(\p)$ which we call $E^o_{\p'}$. At this point we must have $|\mc{M}^o_{\p'}|>0.3251j$ and  therefore $|E^o_{\p'}|<0.1749j+3$. Now there must be an integer $\bar{\ell}$ with $0 \le \bar{\ell} \le 0.1785j+3$ such that $|E^o_{\p'}|=\bar{\ell}$ and 
the number of endpoints within $V(\p)$ (called $E^i_{\p'}$) must be $2(j+1) -\bar{\ell}$. In step 7 we iterate through $\bar{\ell}$. In the next step we cycle through all candidate sets for $E^o_{\p}$ and $E^i_{\p}$ which are called $B_i$ and $B_o$ in the algorithm.\\
In step 9 we consider all possibilities $(e^1_1,e^1_2),\ldots , (e^{j+1}_1,e^{j+1}_2)$ of how to pair the vertices in $B_i \cup B_o$. A pair of endpoints $(e_r^s,e^s_{r+1})$ means that both vertices should appear in the same $P_2$ of $\p'$. Finally, we try to construct $\p'$ from $(e^1_1,e^1_2),\ldots , (e^{j+1}_1,e^{j+1}_2)$ by computing a matching according to \cite{JiaZhaChe2004}.

\subsection{Running Time}
Viewed separately, Algorithm~\ref{algo1} runs in poly-time, as fat and double crowns can be constructed in linear time (Lemmas~\ref{lem3} and~\ref{lem4}). The only exponential run time contribution comes from Algorithm~\ref{algo2}.\\
Here the run times of all steps except steps 3 and 4 are  polynomial in $k$. For any $\ell$ we execute step 3 at most ${3j \choose (j+1)-\ell}  {4j \choose \ell } \in \Oh\left({3j \choose j-\ell}{4j \choose \ell}\right)$ times. Likewise step 8 can be upperbounded by $\Oh \left({3j \choose 2j-\ell}{4j \choose \ell}\right)$.

\begin{lemma}\label{up}
For any integer $z$ with $0\le z \le 0.5j-1$ the following holds:\\[1ex]
$1.$ ${3j \choose j-z}{4j \choose z} < {3j \choose j-(z+1)}{4j \choose z+1}$.\ \ \ $2.$ ${3j \choose 2j-z}{4j \choose z} < {3j \choose 2j-(z+1)}{4j \choose z+1}$.
\end{lemma}

\begin{proof}\
\begin{enumerate}
\item We have
$
\begin{array}{c@{=} c}
 {3j \choose j-(z+1)}{4j \choose z+1} -{3j \choose j-z}{4j \choose z}&\frac{(3j)!(4j)!((j-z)(4j-z) - (2j+z+1)(z+1))}{(j-z)!(2j+z+1)!(z+1)!(4j-z)!}
\end{array} 
$.\\[1ex]
Now it is enough to show $(j-z)(4j-z) - (2j+z+1)(z+1))>0$ which evaluates to $4j^2-7jz-2j-2z-1>0$. For the given $z$ this always is true.\\
\item We have  $
\begin{array}{c@{=} c}
 {3j \choose 2j-(z+1)}{4j \choose z+1} -{3j \choose 2j-z}{4j \choose z}&\frac{(3j)!(4j)!((2j-z)(4j-z) - (j+z+1)(z+1))}{(2j-z)!(j+z+1)!(z+1)!(4j-z)!}
\end{array} 
$.\\[1ex]
Then $((2j-z)(4j-z)-(j+z+1)(z+1))=8j^2-7jz-j-2z-1$  which for the given $z$ is greater than zero.
\qed
\end{enumerate}
\end{proof}
\noindent
With Lemma~\ref{up} step 3  is upperbounded by $\Oh\left({3j \choose (0.6749)j}{4j \choose 0.3251j}\right)$ and step 8 by $\Oh \left({3j \choose 1.8251j}{4j \choose 0.1749j}\right)$. Both are dominated by $\Oh(15.285^j)$. 
 Notice the asymptotic speed-up we achieve by changing the strategy. If we would skip the search for the endpoints, we would have to count $\ell$ up to $0.5j$ in step 3. Then, 
${3j \choose 0.5j}{4j \choose 0.5j} \in \Oh(17.44^j)$ which is also not a big improvement compared to a brute force search for the midpoints on the $7k$-kernel, taking
 $\Oh^*(17.66^k)$ steps. We conclude: 

\begin{theorem}
 {\sc $P_2$-packing} can be solved in time $\Oh^*(2.482^{3k})$.
\end{theorem}

\section{Future work}

It would be nice to derive smaller kernels than $7k$ or $1.5k$ for {\sc $P_2$-packing} or {\sc total edge cover}, resp.,
in view of the mentioned lower bound results~\cite{Cheetal2007}.

We try to apply extremal combinatorial methods to save colors for related problems, like $P_d$-packings for $d\geq 3$.
First results seem to be promising.
So, a detailed combinatorial (extremal structure) study of (say graph) structure under the perspective
of a specific combinatorial problem seems to pay off not only for kernelization (as explained with much
detail in~\cite{Estetal2005}), but also for iterative
augmentation (and possibly compression).

It would be also interesting to work on exact algorithms for {\sc maximum $P_2$-packing}.
By using dynamic programming, this problem can be solved in time $\Oh^*(2^n)$.
By Theorem~\ref{thm-tec-kernel}, {\sc total edge cover} can be solved in time
 $\Oh^*(2^{1.5k})\subseteq\Oh^*(2.829^k)$.
Improving on exact algorithmics would also improve on the parameterized algorithm
for  {\sc total edge cover}. Likewise, finding for example a search-tree algorithm
for  {\sc total edge cover} would be interesting. 

\bibliographystyle{plain}

\end{document}